\documentclass[conference]{IEEEtran}
\IEEEoverridecommandlockouts

\usepackage{cite}
\usepackage{amsmath,amssymb,amsfonts}
\usepackage{graphicx}
\usepackage{textcomp}
\usepackage{xcolor}
\def\BibTeX{{\rm B\kern-.05em{\sc i\kern-.025em b}\kern-.08em
    T\kern-.1667em\lower.7ex\hbox{E}\kern-.125emX}}

\usepackage{amsthm}
\usepackage{hyperref}
\usepackage{prettyref}
\usepackage{booktabs}
\usepackage{tabularx}
\usepackage{tikz}
\usetikzlibrary{arrows.meta,fit,calc,backgrounds}
\usepackage{algorithm}
\usepackage{algpseudocode}
\usepackage{braket}
\usepackage{orcidlink}

\newtheorem{thm}{Theorem}

\newtheorem{definition}{Definition}

\newrefformat{app}{Appendix~\ref{#1}}
\newrefformat{algo}{Algorithm~\ref{#1} on page~\pageref{#1}}
\newrefformat{thm}{Theorem~\ref{#1}}
\newrefformat{cor}{Corollary~\ref{#1}}
\newrefformat{def}{Definition~\ref{#1}}
\usepackage[hyphens]{xurl}

\newcommand{\paren}[1]{\left(#1\right)}
\newcommand{\sqbraket}[1]{\left[#1\right]}
\newcommand{\abs}[1]{\left|#1\right|}

\newcommand\pyreflect{\texttt{PyReflect}}
\newcommand\nixreflect{\texttt{NixReflect}}
\newcommand\Pyreflect{\texttt{PyReflect}}
\newcommand\Nixreflect{\texttt{NixReflect}}

\newcommand\nitrocli{\texttt{nitro-cli}}
\newcommand\nitroutil{\texttt{aws-nitro-util}}

\newcommand\MRENCLAVE{\texttt{MRENCLAVE}}
\newcommand\MRTD{\texttt{MRTD}}
\newcommand\MEASUREMENT{\texttt{MEASUREMENT}}
\newcommand\CONFIGID{\texttt{CONFIG\_ID}}
\newcommand\MARS{\texttt{MARS}}
\newcommand\MAINFO{\texttt{MAINFO}}

\newcommand\render{\texttt{\_\_pyreflect\_render\_\_}}
\newcommand\data{\texttt{\_\_pyreflect\_DATA\_\_}}
\newcommand\self{\texttt{\_\_pyreflect\_SELF\_\_}}

\begin{document}

\title{Bootstrapping Mutual Attestation with Kleene's Second Recursion Theorem}

\author{
    \IEEEauthorblockN{Takuma Imamura\orcidlink{0000-0001-6491-9137}}
    \IEEEauthorblockA{\textit{Confidential Computing Lab} \\
    \textit{Acompany Co., Ltd.}\\
    Nagoya, Japan}
}

\maketitle

\begin{abstract}
Remote attestation is a procedure by which a system (such as a Trusted Execution Environment) proves its own configuration to a remote peer.
Here, the configuration includes the measurements (hash digests) of software components.
When multiple nodes perform mutual attestation, each node must maintain the reference values for its peers' measurements.
We refer to this as the reference-value bootstrapping problem for mutual attestation.
In this paper, we formulate the problem as a fixed-point equation and solve it directly using Kleene's second recursion theorem.
In our approach, nodes can reconstruct one another's source code, thereby allowing each node to derive the reference values at runtime.
When a source code itself is deployed, as with a Python script, a node derives the peer's reference value directly from its reconstructed source.
When a build artifact is deployed, as with AWS Nitro Enclaves, a node instead reproducibly builds the peer's artifact from its reconstructed source and derives the reference value from it.
As a proof of concept, we develop transpilers in Python and Nix to generate mutually referencing programs.
Using these transpilers, we demonstrate the mutual derivation of reference values for the two deployment models mentioned above.
Our solution is architecture-independent, requires no trusted third party, and works with existing attestation stacks without modification.
\end{abstract}


\section{Introduction}
\label{sec:introduction}

\subsection{Background}
\label{sec:background}

\emph{Remote attestation} is a procedure by which a system running in an isolated execution environment proves (attests) its own configuration to a remote peer.
This is a foundational mechanism of confidential computing, a field that has developed rapidly in recent years.
It is particularly important in situations where one party does not trust the others, or where neither party trusts the others.
In such cases, before transmitting confidential data, each node must verify that the peer's system components comply with the expected policy.

The IETF Remote ATtestation procedureS (RATS) architecture \cite{RATS} describes this process in terms of several roles.
An \emph{Attester} measures the target system components, typically by hashing the code loaded into the isolated environment, and emits a signed \emph{attestation report} as evidence.
The calculated hash digests are called \emph{measurements}.
A \emph{Verifier} first validates the evidence using certificates or other endorsement material issued by an \emph{Endorser}, such as the hardware manufacturer.
It then compares the reported measurement with a \emph{reference value}, namely the expected measurement supplied by a \emph{Reference Value Provider}.
The Verifier produces attestation results.
A \emph{Relying Party} uses them to decide whether to trust the Attester.

This standard model is typically one-way.
A Relying Party, such as a remote client or key vault, relies on the appraisal of one or more Attesters before releasing confidential data to them.
The Attesters do not appraise the Relying Party.
In contrast, the subject of this paper is mutual attestation, in which several nodes attest to one another.

Now consider a system in which multiple attestable nodes exchange confidential data with one another.
In such a system, before sending confidential data to another node, each node must verify that node via remote attestation.
Conversely, each node must prove its own configuration before it can receive confidential data.
That is, each node acts as both a Verifier and an Attester.
We refer to this setting as \emph{mutual attestation}.

Such situations naturally arise in systems where multiple confidential computing environments are interconnected.
One example of such a system is a decentralised peer-to-peer confidential computing platform.
Here, physical confidential computing environments are provided by multiple independent providers (often in exchange for cryptocurrency incentives), and workloads are deployed across them.
Other examples include distributed systems comprising multiple confidential computing servers that collaborate with one another, such as the Decent platform \cite{DECENT}, and peer-to-peer confidential computing networks \cite{CarefulWhisper}.
A further example is an agentic AI system comprising multiple AI and tool components running inside separate confidential VMs or containers.
Components supplied by providers that do not trust one another perform mutual attestation before exchanging prompts, contexts, credentials, or intermediate state \cite{AgenTEE}.

In mutual attestation among a multi-node system, each node must be able to access the reference values of its peers.
The most direct way to achieve this is to hard-code the reference value of the attester node into the verifier node.
This approach, however, leads to an infinite regress.
Let us consider two mutually attesting nodes $A$ and $B$, and attempt to hard-code the reference value of node $A$ into the code for node $B$, and \emph{vice versa}.
This will cause changes to the code and measurement values for both $A$ and $B$.
Consequently, the hard-coded reference values must be updated, which in turn causes the measured values to change again.
This continues \emph{ad infinitum}.
We call this infinite regress the \emph{reference-value bootstrapping problem} for mutual attestation.

Existing approaches address or avoid this problem in one of five ways (\prettyref{sec:related}): 
\begin{enumerate}
\item delegating the distribution of reference values to a trusted third party (TTP) \cite{AWSBlog,Teaclave};
\item authenticating externally configured authorisation lists via architecture-specific credentials \cite{DECENT}; 
\item placing configuration in architecture-specific report fields that lie outside the measured scope \cite{KSS};
\item exploiting an architecture-specific measurement process that lets a group of nodes derive one another's reference values \cite{MAGE}; or
\item deploying a common measured image on all nodes, selecting each node's behaviour using externally provisioned role parameters that are bound to the attestation evidence, and using the measurement in locally obtained self-attestation evidence as the common reference value \cite{BMW16}.
\end{enumerate}
The first four approaches either introduce a trusted third party (TTP) or depend on mechanisms specific to a particular isolation architecture.
The fifth avoids both by collapsing all node roles into a single measured identity.
Since every node runs the same image, a node can use its own measurement, obtained through the ordinary attestation mechanism, as the expected measurement of every peer.
This restriction has a practical cost: every node's measured image must include the code and dependencies for every role, even those that the node does not use.

\subsection{Contributions}
\label{sec:contributions}

The aim of this paper is to solve the reference-value bootstrapping problem without a TTP, externally supplied reference values, or architecture-specific measurement mechanisms, while allowing different nodes to retain distinct measured images.
Our approach uses the target architecture's existing software stack and requires no changes to its SDK, signing tools, loader, or attestation mechanism.
Our main contributions are threefold.

\begin{enumerate}
\item We reformulate the reference-value bootstrapping problem as a system of simultaneous fixed-point equations and solve it effectively using the simultaneous form of Kleene's second recursion theorem \cite{Kleene_1938,Smu93}.
It also yields an effective construction that generates a family of mutually referencing programs from a finite set of program templates.
We implement this general construction in the \pyreflect{} and \nixreflect{} transpilers for Python and Nix, respectively (\prettyref{sec:mutual-referencing-programs}).

\item For source-code deployments (e.g. Python scripts measured by Linux IMA), each node reconstructs a peer's exact source code and applies the corresponding source-code measurement procedure.
As a proof of concept, we instantiate this case with a two-node Python system generated by \pyreflect{}, in which the nodes perform post-handshake TPM-backed mutual attestation, each verifying its peer against a reference value it computes itself (\prettyref{sec:tpm}).

\item For build-artifact deployments (e.g. SGX enclaves, confidential VMs, and AWS Nitro Enclaves), the construction uses \emph{reproducible builds} \cite{Reproducible}: a node must rebuild its peer's deployed artifact bit-for-bit before applying the architecture's measurement procedure.
As a PoC, we instantiate this case with a \nixreflect{}--generated system of two AWS Nitro Enclaves, each of which rebuilds its peer's EIF bit-for-bit and derives the corresponding reference PCR0--2 values (\prettyref{sec:nitro}).
\end{enumerate}

Using our construction, each verifier can derive its peer's reference value within its own TCB, and uses it to appraise the peer, independently in each direction.
Therefore, the mutual attestation based on our method fully inherits the trust assumptions, threat model, and security guarantees of the underlying isolation architecture and one-way-attestation protocol (\prettyref{sec:security}).

In \prettyref{sec:cost}, we estimate the cost of transpilation.
We also compare two variants of the \emph{mutual quine} to measure the runtime cost of enclave rebuilding and reference-value derivation.
One reconstructs and simply hashes the peer's Nix source, while the other reproducibly rebuilds the peer's enclave image and derives its reference PCRs.
In our PoC, the source-to-image-to-PCR computation is approximately $100\times$ slower than hashing alone; because rebuilding traverses the workload's full build closure, this gap is expected to widen for more complex applications.

In \prettyref{sec:discussion}, we examine two operational constraints of thr current design and implementation.
First, attestable builds \cite{AttestableBuilds} can eliminate the cost of rebuilding at runtime without trusting the party that issues the build certificate.
However, a relying party must trust the CCP's hardware root and endorsement PKI.
This is an additional trust assumption when the CCP and deployed nodes are heterogeneous.
Secondly, with the current implementation, changes to a node definition propagate throughout the system, making updates expensive.
We provide a more efficient transpilation method, making it possible to limit the scope of propagation to only those nodes that are truly dependent.
The method is based on the topological order within the condensation graph of the node dependency graph.
We also provide a proof of its correctness.

In \prettyref{sec:related}, we compare our construction with the five existing approaches outlined above and summarise their similarities and differences.

\section{Mutually referencing programs}
\label{sec:mutual-referencing-programs}

This section describes a method for constructing mutually referencing programs, independently of remote attestation.
Given a finite family of program templates that may refer to the source code of any member of the family, the recursion theorem effectively produces a corresponding family of mutually referencing programs.
We first state the underlying fixed-point result and then show how the \pyreflect{} and \nixreflect{} transpilers implement it.

We refer to Odifreddi~\cite{Odi89} for terminology and notation in recursion theory.

\subsection{Simultaneous fixed points}
\label{sec:mutual-fixed-points}

Throughout this section, we model a group of $k$ nodes with two assumptions.
\begin{enumerate}
\item[(i)] \emph{Programs}. A node is identified with (the G\"odel number of) its source code $e$ and computes the partial function $\varphi_{e}$. Here $\varphi$ is a fixed acceptable numbering (\prettyref{def:numbering} below), which serves as the denotational semantics of the implementation language.
\item[(ii)] \emph{Source access}. The intended behaviour of node $i$ is a computable partial function $f_{i}\paren{ \vec{e}, {-} }$.
During execution it may use the source code $\vec{e} = \paren{ e_{1}, \ldots, e_{k} }$ of every member.
\end{enumerate}

Under these assumptions, a mutually referencing family is exactly a tuple $\vec{e}$ of source codes solving the system of equations:
\begin{equation}
\label{eq:fixpoint}
\varphi_{e_{i}} = f_{i}\paren{ e_{1}, \ldots, e_{k}, {-} }, \qquad 1 \leq i \leq k.
\tag{$\dagger$}
\end{equation}
Here, the unknowns $e_{i}$ occur on both sides.
The left-hand side is the partial function computed by program $e_{i}$.
The right-hand side $f_{i}\paren{ \vec{e}, {-} }$ specifies node $i$'s behaviour when the tuple $\vec{e}$ of node source codes is supplied as an additional parameter.

A hard-coding approach repeatedly substitutes one definition into another and, when the dependency graph contains a cycle, never terminates\footnote{More precisely, consider the directed graph $G := \paren{V\paren{G}, E\paren{G}}$, where $V\paren{G} := \set{i|1\leq i\leq k}$ and $E\paren{G} := \set{\paren{i, j} | f_{i} \text{ depends on variable } e_{j}}$.
If $G$ is acyclic, hard-coding works as expected.
However, if $G$ contains a directed cycle, hard-coding results in an infinite regress along that directed cycle.
See also \prettyref{sec:updates}.
}.
The system as a whole, however, is simply a \emph{simultaneous fixed-point equation}, which recursion theory solves in full generality.

\begin{definition}\label{def:numbering}
Let $\mathcal{P}^{\paren{n}}$ be the class of all computable partial functions of type $\mathbb{N}^{n} \rightharpoonup \mathbb{N}$.
A surjection $\varphi \colon \mathbb{N} \to \mathcal{P}^{\paren{n}}$ is called a \emph{computable numbering} of $\mathcal{P}^{\paren{n}}$ if the $\paren{n+1}$--ary partial function $\paren{e, x} \mapsto \varphi_{e}\paren{x} := \varphi\paren{e}\paren{x}$ is computable.
A computable numbering $\varphi$ of $\mathcal{P}^{\paren{n}}$ is said to be \emph{acceptable} if for each computable numbering $\psi$ of $\mathcal{P}^{\paren{n}}$, there exists a computable function $r\colon \mathbb{N} \to \mathbb{N}$ (called a \emph{many-one reduction}) such that $\psi = \varphi \circ r$.
\end{definition}

Intuitively, a computable numbering is an interpreter (or evaluation function) for a particular programming language; $\varphi_{e}$ is the partial function computed by code $e$; the many-one reduction $r$ is a transpiler from language $\psi$ to language $\varphi$.
An acceptable numbering can be thought of as a universal programming language into which any other language can be effectively translated.
Without a special construction such as Friedberg's priority argument \cite{Fri58}, any reasonable Turing-complete programming language naturally induces an acceptable numbering.

The single-node case ($k=1$) of system~\eqref{eq:fixpoint} is precisely the subject of Kleene's second recursion theorem.

\begin{thm}[Kleene~\cite{Kleene_1938}]
\label{thm:kleene}
For each $n\in\mathbb{N}$, let $\varphi^{\paren{ n }}$ be an acceptable numbering of $\mathcal{P}^{\paren{n}}$.
For every computable partial function $f \colon \mathbb{N}\times\mathbb{N}^{n} \rightharpoonup \mathbb{N}$, there exists an $e \in \mathbb{N}$ such that
\[
\varphi^{\paren{n}}_{e} = f\paren{ e, {-} }.
\]
Moreover, the index $e$ can be chosen effectively: there exists a computable function $e\colon \mathbb{N} \to \mathbb{N}$ such that
\[
\varphi^{\paren{n}}_{e\paren{i}} = \varphi^{\paren{n+1}}_{i}\paren{ e\paren{i}, {-} }.
\]
\end{thm}

Mutual reference among several nodes ($k > 1$) requires the simultaneous form of the theorem, which follows from the single-node case.

\begin{thm}[Smullyan~\cite{Smu93}]
\label{thm:simultaneous}
Let $\varphi^{\paren{n}}$ be an acceptable numbering of $\mathcal{P}^{\paren{n}}$.
For every $k$-tuple of computable partial functions $f_{1}, \dots, f_{k} \colon \mathbb{N}^{k} \times \mathbb{N}^{n} \rightharpoonup \mathbb{N}$ there exists a $k$--tuple of indices $\vec{e} := \paren{ e_{1}, \ldots, e_{k} } \in \mathbb{N}^{k}$ such that
\[
\varphi^{\paren{n}}_{e_{i}} = f_{i}\paren{ \vec{e}, {-} }, \qquad 1 \leq i \leq k.
\]
Moreover, the tuple $\paren{ e_{1}, \ldots, e_{k} }$ can be computed effectively from indices of the $f_{i}$.
\end{thm}

The two clauses of \prettyref{thm:kleene} do different work.
The first half is an \emph{existence} statement: for every behaviour $f$ there is a program $e$ that may use its own source code, so self-reference never stands in the way of implementation.
The second \emph{effective} clause strengthens the existence result by providing an effective construction.
In programming terms, this algorithm is a transpiler that takes any program expecting its own source code as an extra input and turns it into a genuinely self-referential program with the same behaviour.

The same reading applies to \prettyref{thm:simultaneous}: an entire mutually referencing family can be generated automatically from a template of behaviours written with the solution codes as free variables.
This forms the theoretical basis of our transpilers \Pyreflect{} and \Nixreflect{} (\prettyref{sec:transpiler}).
These generate mutually referencing programs by mimicking the construction of the map $\paren{f_{1}, \ldots, f_{k}} \mapsto \paren{e_{1}, \ldots, e_{k}}$.
We write out the proofs of these theorems in \prettyref{app:proofs}, because the transpilers follow the same self-application pattern and are best read alongside them.

\subsection{\Pyreflect{} and \nixreflect}
\label{sec:transpiler}

We developed the \pyreflect{} transpiler to implement the fixed-point construction for Python.
It converts a tuple of mutually referencing Python pseudocode bodies into a tuple of well-formed Python source files that refer to one another.

Its input is a JSON list of entries, called a \emph{template}.
Each entry consists of a \emph{node-id} $N$ and a \emph{body} $P_{N}$ of Python pseudocode, which may contain node-ids as undefined symbols.
An occurrence of a node-id $N'$ inside a body stands for the source code of node $N'$.
A template thus represents the tuple $\paren{f_{1}, \ldots, f_{n}}$ in equation~\prettyref{eq:fixpoint}, with each intended behaviour written with the solution codes as free variables.
Its output is a solution to equation~\prettyref{eq:fixpoint}, represented by one well-formed Python source file for each node-id.
At runtime, each emitted node reconstructs the exact source code of every family member, byte for byte, using only built-in (intrinsic) data.

The implementation of this transpiler mimics the proof of \prettyref{thm:simultaneous} in the Python language.
The construction has two key steps.
First, the family of nodes is treated as a \emph{single} system.
The whole template is serialised into one canonical data blob (Base64-encoded canonical JSON), which is shared across all nodes.
Each emitted node consists of that blob, auxiliary functions for reconstructing the source from that blob, and the node body.
Mutual reference among $n$ programs thereby reduces to self-reference of one system.
This is the same reduction used to derive \prettyref{thm:simultaneous} from \prettyref{thm:kleene}.

Second, the self-reference is implemented using Kleene's trick at runtime.
Each occurrence of a node-id in the node body is rewritten into a call to a reconstruction function.
This function reassembles any member's source code from the embedded blob.
\prettyref{app:design} shows the details of the emitted code layout, the reconstruction function, and the correspondence with the proof of Kleene's theorem (\prettyref{app:proofs}).
It is worth noting that the same trick is used in the classic `Quine' program, which outputs its own source code.

The rewriting mechanism is tied neither to Python nor to hashing: it applies to any language whose programs are to be reproduced verbatim from built-in data.
\Nixreflect{} (itself written in Python) takes a family of mutually referencing Nix expressions \cite{Nix} and turns them into nodes, each able to reconstruct the exact build definition of every member by Kleene's trick.
Both transpilers implement the same effective fixed-point construction.
They differ only in the language being transformed and in the downstream use of the reconstructed source.

\section{Reference value derivation for mutual attestation}
\label{sec:reference-value-derivation}

In this section, we apply the results of \prettyref{sec:mutual-referencing-programs} to realise mutual attestation.
For this application, we add the following third assumption:
\begin{enumerate}
\item[(iii)] \emph{Computable measurement}. A node's measurement (hence its reference value) is a fixed (deterministic) computable function of its source code.
Having access to a peer's source code therefore suffices to compute its reference value.
\end{enumerate}

In this application, the intended behaviour $f_{i}$ includes the node's attestation behaviour and may use each peer's source to derive the corresponding reference value.
Under this interpretation, \prettyref{thm:simultaneous} states that system~\eqref{eq:fixpoint} always has a solution that can be computed effectively from the behaviour functions.
Under assumptions~(i)--(iii), reference values can therefore be constructed without a trusted third party or a host-injected measurement.
Each node $e_i$ can compute the exact code (and hence, by assumption~(iii), the reference value) of every node in the family.

Assumptions~(i) and (ii) concern the program generation layer and are independent of the architecture or deployment method.
Assumption~(iii) is substantive, because it depends on how source code is transformed into the deployed representation covered by the architecture's measurement.

\subsection{Trust assumptions and security considerations}
\label{sec:security}

Our construction relies on the same trust assumptions as the conventional case in which a Verifier itself holds the reference values that it uses.
In RATS terms, it realises a deployment in which the Reference Value Provider and Verifier roles are combined: the Verifier derives and retains the expected values within its own TCB, rather than obtaining them from an external provider.
The construction therefore introduces no additional trusted party.

All remaining trust assumptions and security guarantees are inherited from the underlying isolation architecture and its one-way-attestation protocol.
In our setting, mutual attestation consists of two independent one-way attestations, one in each direction.
Each attestation uses a reference value that the corresponding Verifier has derived beforehand.
Our construction neither strengthens nor weakens these architecture-specific guarantees.
It only supplies the reference values needed to invoke them without an external Reference Value Provider.
Cryptographic assumptions on the measurement function likewise belong to the underlying architecture.
They are not assumptions of the fixed-point construction itself.

\subsection{From reconstructed source to reference measurements}
\label{sec:source2meas}

The fixed-point construction gives each node access to the exact source of every peer.
How that source is converted into a reference value depends on the form in which the peer is deployed.
In general, let $\beta$ be the procedure that transforms source code $e$ into the build artifact $\beta\paren{e}$ to be deployed, and let $\mu$ be the architecture-specific measurement function over that deployed artifact.
Then the source-to-reference correspondence is given by their composition $\mu \circ \beta$.

There are two forms of software deployment.
In the first form, the source code itself is deployed to the target system and executed.
The code is typically written in an interpreted language.
We call this form \emph{source-code deployment} or \emph{interpreting deployment}.
In the second form, a build artifact generated from the source code is deployed instead of the source code itself.
We call this form \emph{build-artifact deployment}.
It covers software implemented in compiled languages, VM images, and container images.
The derivation of reference values differs between the two forms.

In source-code deployment, the source file itself is deployed and measured.
No separate build procedure is required, or in other words, $\beta$ is equal to the identity map $\operatorname{id}$.
Therefore, to obtain the reference value, it is sufficient to apply $\mu$ directly to the source code.
On the other hand, in build-artifact deployment, $\beta$ is the build procedure that produces the build artifact.
The artifact is then deployed and measured.

A server script running on a Linux platform with a Truted Platform Module (TPM) device is a typical example of the former deployment model \cite{IMA, TPM}.
The Integrity Measurement Architecture (IMA) is an integrity subsystem of the Linux kernel that protects system integrity by performing measurement (and appraisal).
IMA measures files in accordance with a measurement policy, appends the measurement results to a log, and records the log's hash chain in the Platform Configuration Registers (PCRs) of the platform TPM.
The TPM can then issues a signed evidence (called a TPM quote) with the relevant PCR values.
Thus, when a server script is loaded, IMA measures that script, and the TPM signs the resulting PCR state, incorporating this as part of the attestation evidence.

On the other hand, most Trusted Execution Environments (TEEs) fall into the latter deployment model.
In Intel SGX, software is developed with the official SGX SDK or a library OS such as Gramine~\cite{Gramine} or Occlum~\cite{Occlum}, and is then built into an enclave image.
Its measurement (\MRENCLAVE{}) is computed as an incremental hash over the enclave's initial pages in load order \cite{SGX}.
Similarly, a confidential VM's launch measurement covers the initial guest memory state rather than the source files that define the VM image.
Examples of such measurements include \MRTD{} on TDX and the launch \MEASUREMENT{} on SEV-SNP \cite{TDX, SNP}.
In confidential containers, the container policy includes the measurements of permitted container images, and the policy digest is bound to the attestation report.
Likewise, AWS Nitro Enclaves provide hypervisor-based isolation and attestation.
There, enclave image files are built from Dockerfiles and are then deployed and measured.

Assumption~(iii) requires the source-to-measurement correspondence $e \mapsto \mu\paren{\beta\paren{e}}$ to be a deterministic, computable function.
For build-artifact deployment, this requirement is fulfilled by using reproducible builds \cite{Reproducible}.

The fixed-point construction itself is independent of the deployment model and the isolation architecture.
Indeed, even if these were to change, the source reconstruction step would remain the same; only the measurement function $\mu \circ \beta$ would change.
In this specific sense, our method is architecture-independent.

\begin{table}
\caption{Architecture-specific measurement paths, grouped by deployment form.}
\label{tab:source2meas}
\small
\setlength{\tabcolsep}{3pt}
\renewcommand{\arraystretch}{1.12}
\begin{tabularx}{\linewidth}{
  @{}
  >{\raggedright\arraybackslash}p{0.18\linewidth}
  >{\raggedright\arraybackslash}X
  >{\raggedright\arraybackslash}X
  >{\raggedright\arraybackslash}X
  @{}
}
\toprule
Architecture
  & Measured object
  & Measurement record
  & Evidence binding the record \\
\midrule

\multicolumn{4}{@{}l}{\emph{Source-code deployments}} \\
\addlinespace[0.25em]

Linux IMA (TPM-backed)
  & Policy-selected deployed source files and selected metadata
  & Per-file entries in the IMA measurement log; hash chain in PCR~10
  & TPM quote over PCR~10, checked together with a replay of the IMA log \\

TPM PoC (\prettyref{sec:tpm})
  & Deployed Python source file
  & SHA-256 digest extended into PCR~23
  & TPM quote over PCR~23 \\

\addlinespace[0.5em]
\multicolumn{4}{@{}l}{\emph{Build-artifact deployments}} \\
\addlinespace[0.25em]

Intel SGX
  & Initial enclave pages and load metadata
  & \MRENCLAVE{}
  & SGX quote containing \MRENCLAVE{} \\

Intel TDX
  & Initial guest TD state
  & \MRTD{}
  & TD quote containing \MRTD{} \\

AMD SEV-SNP
  & Initial guest CVM state
  & Launch \MEASUREMENT{}
  & SNP attestation report containing \MEASUREMENT{} \\

Confidential containers (policy-based)
  & OCI container image, identified by its digest
  & Image digest recorded in the container policy; digest of the
    Init-Data containing that policy
  & Attestation report binding the attested Init-Data digest \\

AWS Nitro Enclaves
  & EIF components: image, kernel and bootstrap, and application
  & PCRs~0--2
  & Attestation document containing the PCRs
    (\prettyref{sec:nitro}) \\

\bottomrule
\end{tabularx}
\end{table}

\prettyref{tab:source2meas} summarises typical examples of both deployment models.
In the following subsections, we discuss these cases in detail.

\subsection{Source-code deployments: TPM-backed mutual attestation}
\label{sec:tpm}

We first instantiate the construction for a source-code deployment, in which the source file itself is deployed and the measurement procedure is applied directly to that file.

In this deployment model, there is no need to build the reconstructed code to derive the reference values.
Each node only needs to apply the measurement function $\mu$ to the peers' reconstructed source code.
The general construction proceeds as follows:
\begin{enumerate}
\item[(1)] Before deployment, transpile the family of nodes using \prettyref{thm:simultaneous}, so that each node can reconstruct every peer's exact source code.
\item[(2)] At run time, reconstruct the peer's source code and apply the source-code measurement procedure $\mu$, obtaining the peer's reference value.
\item[(3)] Run the architecture's standard attestation flow, checking the measurement reported in the peer's evidence against the self-computed reference value.
\end{enumerate}

As a proof of concept, we implement a pair of Python scripts that perform TPM-backed mutual attestation after Elliptic-Curve Diffie--Hellman (ECDH) key exchange (see, e.g., \cite{TPM,aTLS}).
The pair is generated by \pyreflect{}.
Reference values are self-computed: a node reconstructs its peer's exact source with the reconstruction function and takes its SHA-256 digest.
Because the reconstructed source is precisely the deployed source, this digest coincides with the measurement of the peer's deployed source file.

\emph{Protocol flow.} In our PoC, two nodes, $A$ and $B$, first perform an ECDH key exchange, then exchange TPM quotes that bind the shared session key and verify each other's quotes.
More precisely, consider one direction, in which node $A$ verifies node $B$; the reverse direction executes the same steps independently.
For the sake of simplicity, PCR~23 is used as the measurement register.
\begin{enumerate}
\item $A$ generates an ephemeral P-256 ECDH key pair $\paren{\mathit{pub}_{A}, \mathit{priv}_{A}}$ and a nonce $n_{A}$.
\item $A$ reconstructs $B$'s exact source, computes its SHA-256 digest, and replays the PCR reset-and-extend operation to obtain the reference PCR value $r_{B}$.
\item $A$ and $B$ exchange their nonces and ephemeral ECDH public keys.
\item $A$ and $B$ compute the ECDH shared secret $Z$ and derive two keys $\mathit{VK}$ and $\mathit{SK}$ from $Z$ using HMAC-SHA256.
\item $B$ creates an ECDSA attestation key (AK) in its TPM, resets the PCR, extends the SHA-256 digest of its own on-disk source code into that PCR, and requests a TPM quote with $\mathtt{SHA256}\paren{n_{A} \parallel \mathit{pub}_{B} \parallel \mathit{pub}_{A} \parallel \mathit{VK} }$ as qualifying data.
\item $B$ returns its AK public key, PCR value, quote, and quote signature.
\item $A$ verifies the signature under the received AK, validates the TPM attestation structure, checks that its \texttt{extraData} field equals $\mathtt{SHA256}\paren{n_{A} \parallel \mathit{pub}_{B} \parallel \mathit{pub}_{A} \parallel \mathit{VK} }$, and recomputes the quoted \texttt{pcrDigest} from the returned PCR value.
It accepts $B$ only if the returned PCR value equals $r_{B}$.
\item After both directions have succeeded, the nodes exchange messages encrypted with AES-256-GCM under $\mathit{SK}$.
\end{enumerate}
The only difference from the usual flow is the derivation of $r_{B}$.
Rather than receiving $r_{B}$ from an external Reference Value Provider, $A$ derives $r_{B}$ from $B$'s reconstructed source.

This PoC makes two simplifications.
First, under the default settings, each node uses a separate software TPM emulator to ensure reproducibility and avoid affecting the host platform's TPM.
Since these TPM emulators operate in userland, they do not provide a hardware root of trust.
For example, the credential activation protocol would show only that an AK belongs to the same emulator instance as its emulator-generated endorsement key (EK); it would not establish a hardware root of trust.
Therefore, this PoC does not perform verification of the AK public key (such as the AK--EK binding), but simply trusts each received AK.
Note that it is also possible to run the PoC using a vTPM \cite{vTPM} isolated within a TEE or a physical TPM, which provide a hardware root of trust.

Second, each node running in userland computes its own code measurement and extends the resulting digest into the resettable PCR~23.
Thus, even a valid quote over PCR~23 shows only the PCR state at the time of quoting.
It does not establish that the extended digest was computed from the Python script actually running.
In a production deployment, both the measurement and the PCR extension should instead be performed by a mechanism within the Trusted Computing Base (TCB), such as Linux IMA \cite{IMA}, and the measurement should be extended into a PCR that cannot be reset from outside the TCB, such as PCR~10.

However, these simplifications are unrelated to the feasibility of mutual attestation or reference value derivation.
There is nothing to prevent the integration of our method into rigorous TPM attestation protocols.

\subsection{Build-artifact deployments: Nitro Enclaves}
\label{sec:nitro}

We next instantiate the same construction for a build-artifact deployment.
In this case, source code $e$ is first transformed by a build procedure $\beta$ into an executable artifact $\beta\paren{e}$, which is then deployed and measured.
The build procedure $\beta$ is not generally deterministic in practice.
This non-determinism arises from factors such as unpinned build dependencies, variable timestamps and metadata, and non-deterministic ordering during static linking.

This deployment model therefore requires additional build steps.
\begin{enumerate}
\item[(1)] Before deployment, transpile the family of nodes exactly as above, and \emph{build the artifacts} to be deployed from the generated source code.
\item[(2)] At run time, reconstruct the peer's source code, \emph{reproducibly rebuild the peer's artifact so that all measurement-relevant bytes are bit-for-bit identical to those of the deployed artifact,} and apply the artifact-measurement procedure $\mu$, obtaining the peer's reference value.
\item[(3)] Run the architecture's standard attestation flow, checking the measurement reported in the peer's evidence against the self-computed reference value.
\end{enumerate}

The requirement in step~(2) is essential: artifacts built from the same source code $e$ must always have the same measurement $\mu \circ \beta\paren{e}$, as required by assumption~(iii).
As noted above, however, an ordinary build procedure may violate this requirement.
We meet this requirement with \emph{reproducible builds} (see e.g. \cite{Reproducible}).
By pinning all build dependencies by hash and using a reproducible build toolchain, we ensure that the same source code always produces the same build artifact, and hence the same measurement.

As a proof of concept for this deployment model, we implement a variant of the \emph{mutual quine} on AWS Nitro Enclaves, in which two enclaves display each other's reference values.
This PoC is generated by \nixreflect{}.

For Nitro Enclaves, the build-derived measurements relevant here are PCR0--2, exposed when the enclave image file (EIF) is built \cite{AWSNitroEnclaves}.
In the standard toolchain, an enclaves is defined via a Dockerfile and built into an EIF using AWS's official \nitrocli{}.
However, this build pipeline is generally not reproducible \cite{DockerNonRepro}.
To address this, we instead use Nix~\cite{Nix} to define an enclave and Monzo's \nitroutil{} \cite{nitro-util} to build an EIF reproducibly.

\emph{Reference-value derivation.}
Again, consider one direction, in which enclave $A$ derives the expected measurements of enclave $B$.
\begin{enumerate}
\item $A$ reconstructs $B$'s exact enclave definition represented as a Nix expression.
\item $A$ reproducibly rebuilds $B$'s EIF inside $A$ itself.
\item $A$ applies the EIF measurement procedure to the reproduced EIF and obtains the reference tuple $r_{B}=\paren{r_{B,0},r_{B,1},r_{B,2}}$ for PCRs~0--2.
\end{enumerate}
The tuple reconstructed inside $A$ matches the PCR0--2 values recorded when the EIF deployed as $B$ was built, and symmetrically for $B$ and $A$.
The nodes thus obtain one another's reference PCRs without a TTP, thereby resolving the reference-value bootstrapping problem which AWS addresses via a measurement notary (\prettyref{sec:notary}).

\emph{Integration to Nitro Enclaves attestation.}
The tuple $r_{B}$ can be used in standard attestation protocols.
For example, a post-handshake mutual attestation process proceeds as follows \cite{AWSNitroEnclaves,aTLS}:
\begin{enumerate}
\item $A$ and $B$ perform a key exchange and derive the shared secret $Z$.
\item $A$ sends a fresh nonce $n_{A}$ to $B$.
\item $B$ requests a signed attestation document from the Nitro Secure Module, placing $n_{A}$ in the document's \texttt{nonce} field and the session information (such as the derived key from $Z$) in \texttt{user\_data}.
\item $B$ returns the attestation document to $A$.
\item $A$ decodes its CBOR/COSE structure, validates the certificate chain and signature under the AWS Nitro Attestation PKI, checks the nonce and session binding, and extracts the reported tuple $m_{B}$ from PCR0--2.
\item After performing any additional deployment-specific appraisal, $A$ accepts $B$'s software identity only if $m_{B}$ matches the independently derived reference tuple $r_{B}$.
\end{enumerate}
The reverse direction executes the same flow independently with the roles exchanged.

The PoC demonstrates only the preceding reference-value-derivation steps; it does not reimplement the full post-handshake mutual-attestation protocol.
Accordingly, we omit the attestation-document exchange and validation steps described in the second sequence.

\subsection{A note on intrinsic verifiability}
\label{sec:intrinsic}

We close this section with an observation that distinguishes our construction from the alternatives discussed in \prettyref{sec:related}.
The guarantees it provides are \emph{intrinsic} to the measured artifacts.
Every input to reference-value derivation is contained in the node's own measured code, and no value is provisioned into the node from outside after measurement.

This has two consequences.
First, a node can verify its own deployment.
It reconstructs its own source, applies the correspondence $\mu \circ \beta$ of \prettyref{sec:source2meas}, and compares the result with the measurement in its own attestation evidence.
Agreement confirms that the deployed artifact is the one described by the embedded representation.
This check detects build non-determinism and deployment mismatches rather than defending an already compromised node.
Second, a node can inspect each peer's reconstructed source and confirm that the peer appraises it in turn.
The fact that the nodes mutually attest, and against which code, is thus verifiable from within the family without external information.
The same holds for any auditor who obtains the source of a single node, since that source reconstructs the entire family.

By contrast, the existing approaches that provision peer reference values from outside the node leave a node unable to judge those values itself.
This is so whether the values come from a trusted third party or from an untrusted host whose values are only checked for consistency among peers.
We return to this distinction in \prettyref{sec:related}.

\section{Cost}
\label{sec:cost}

This section considers two distinct costs of our construction.
First, we analyse the size of the source code emitted by the transpiler as a function of the input template size and the number of nodes.
Second, we evaluate the run-time cost of deriving a peer's reference value.

At run time, reference-value derivation consists of reconstructing the peer's source and mapping that source to the measurement.
For a source-code deployment (\prettyref{sec:tpm}), the reconstructed source is itself the deployed representation, so the measurement procedure can be applied directly to it.
For a build-artifact deployment (\prettyref{sec:nitro}), the node must first reproducibly rebuild the peer's deployed artifact and then apply the artifact-measurement procedure.

\subsection{Transpilation}

As detailed in \prettyref{app:design}, \pyreflect{} produces, for each node $N$ in the input template, source code $e_{N}$ of the form
\[
e_{N} = H \parallel \operatorname{SELF}\paren{N} \parallel \operatorname{DATA}\paren{d} \parallel F \parallel \widehat{P}_{N}.
\]
Here $H$ is a fixed Python header shared by all generated nodes.
The block $\operatorname{SELF}\paren{N}$ embeds the node-id $N$ of the generated source.
The block $\operatorname{DATA}\paren{d}$ embeds the Base64 encoding of a canonical JSON record containing the set $\mathcal{N}$ of node-ids, all rewritten bodies $\widehat{P}_{M}$ ($M\in\mathcal{M}$), $H$ and $F$.
The block $F$ is a fixed framework for code reconstruction.
It includes functions for decoding and accessing the embedded data $d$, together with a \texttt{render} function that reconstructs the complete emitted source $e_{M}$ from $d$ and an input node-id $M$.
Finally, $\widehat{P}_{N}$ is the body $P_{N}$ in which each reference to a node-id $M$ has been replaced by a call of \texttt{render} on $M$.

Now, let $T$ be the total byte size of the input template, including its node-ids and bodies, and let $k$ be the number of nodes.
Rewriting the node-id occurrences, serialising the canonical record, and Base64-encoding the record increase the record's size by at most a constant factor.
Since $H$ and $F$ have constant size and $d$ contains all rewritten bodies and node-ids, we have
\begin{gather*}
\abs{d} = \Theta\paren{T}, \qquad
\abs{e_{N}} = \Theta\paren{T}, \qquad
\sum_{N\in\mathcal N}\abs{e_N} = \Theta\paren{kT}.
\end{gather*}

For example, in the Nitro Enclaves PoC, a two-node template of $0.9$\,KB yields nodes of $3.4$\,KB each.

\subsection{Source-to-measurement derivation}

Unlike the generated code size, the runtime cost is not fixed by the construction.
We instead measure it using two variants of mutual quine generated by \nixreflect{}.
The \emph{rebuild} variant is precisely the PoC of \prettyref{sec:nitro}: two Nitro Enclaves compute each other's reference PCR values.
The \emph{digest} variant has the same layout but computes the SHA-384 digest of the reconstructed source rather than PCR values, thereby isolating the cost of the reproducible build.

\prettyref{tab:bench-env} summarises the configuration relevant to interpreting the comparison.
See \prettyref{app:bench-env} for more details.
The two variants run on the same Nitro Enclaves host VM under identical resource allocations.
The EIF size of the digest variant is approximately $42$\,MB per node, while that of the rebuild variant is $50$\,MB per node.
For each run, we measure the elapsed time required inside the enclave to derive the peer's reference value.
For the digest variant, this interval covers source reconstruction and hashing.
For the rebuild variant, it additionally covers the reproducible EIF build and PCR derivation.
We report enclave start-up latency separately, since it is a platform-level cost common to both variants rather than part of reference-value derivation.

\begin{table}
\caption{Benchmark host and enclave configuration.}
\label{tab:bench-env}
\small
\begin{tabularx}{\columnwidth}{@{} l >{\raggedright\arraybackslash}X @{}}
\toprule
Item & Configuration \\
\midrule
Host & \texttt{m6g.xlarge}: AWS Graviton2 processor (Arm Neoverse-N1), 4~vCPUs, 1~thread/core, 16~GiB RAM; Nitro Enclaves enabled \\
OS & Ubuntu 26.04 LTS; Linux 7.0.0-1010-aws \\
Per-enclave allocation & $2$ vCPUs and $2048$\,MiB; one enclave at a time, identical for both variants \\
Principal tools & \nitrocli{} 1.4.5; Determinate Nix 3.21.8 (Nix 2.34.8), flakes enabled \\
\bottomrule
\end{tabularx}
\end{table}

\begin{table}
\caption{Deriving a peer's reference value inside a Nitro Enclave with \nixreflect{}:
\texttt{mutual\_quine\_ne\_sha} (source digest) vs.\ \texttt{mutual\_quine\_ne\_pcrs} (reproducible rebuild).
Mean $\pm$ sample standard deviation over 10 runs for each node; observations from the two nodes are not pooled.}
\label{tab:cost}
\small
\begin{tabular}{@{} l c r r @{}}
\toprule
Metric & Node & digest & rebuild \\
\midrule
derivation (s) & 1 & $0.069 \pm 0.007$ & $6.853 \pm 0.024$ \\
 & 2 & $0.064 \pm 0.007$ & $6.839 \pm 0.017$ \\
launch $\to$ entrypoint (s) & 1 & $1.486 \pm 0.036$ & $1.561 \pm 0.042$ \\
 & 2 & $1.364 \pm 0.033$ & $1.684 \pm 0.038$ \\
launch $\to$ result (s) & 1 & $1.554 \pm 0.036$ & $8.415 \pm 0.042$ \\
 & 2 & $1.429 \pm 0.032$ & $8.524 \pm 0.048$ \\
\bottomrule
\end{tabular}
\end{table}

\prettyref{tab:cost} shows that deriving a reference value by rebuilding takes approximately two orders of magnitude longer than hashing.
For this PoC, the absolute cost is modest: under seven seconds to re-stage the peer's rootfs, re-pack its ramdisk, and re-run the EIF assembly in RAM, and under ten seconds from cold launch to the peer's reference PCRs.
This result reflects a simple enclave whose application-level task is only to rebuild its peer and derive the peer's reference PCRs.

Both the rebuild time and the image size grow with the node's closure.
For more complex enclaves, deriving a reference value is therefore likely to take longer because rebuilding traverses a larger build closure.
A larger build may also require more enclave memory to hold the toolchain, inputs, and intermediate artifacts.
Both effects make it more attractive to offload the runtime build.
We discuss this option in \prettyref{sec:offload}.

The digest variant's sub-$0.1$\,s derivation time is close to the resolution of our measurement method and should therefore be interpreted only as negligible relative to the rebuild cost.
To observe the beginning and end of the in-enclave computation, the experiments use fresh debug-mode enclaves and derive elapsed times from the corresponding console events observed by the parent VM.
This instrumentation is adequate for distinguishing the approximately $7$\,s rebuild from the digest-only computation, but not for making a precise claim about the latter's absolute latency.

\section{Discussion}
\label{sec:discussion}

In this section, we examine two operational constraints of the current design and implementation: the runtime cost associated with rebuilding peer artifacts and the tight coupling of node updates imposed by the current transpiler.
We then demonstrate that attestable builds and incremental transpilation utilising topological sorting can mitigate these respective constraints.

\subsection{Offloading the runtime build}
\label{sec:offload}

For build-artifact deployments (\prettyref{sec:nitro}), each node computes reference values by carrying out a reproducible build at runtime.
This is the dominant overhead of the construction: the node must carry a full build toolchain and pay the peer's build cost. In the Nitro Enclaves PoC, the enclave image must include Nix and all pinned inputs.

An alternative is to perform the source-to-measurement step before deployment.
\emph{Attestable builds} \cite{AttestableBuilds} run a build inside a TEE-backed sandbox and emit a hardware-rooted \emph{confidential computing proof} (CCP) \cite{CCP} that a specific artifact, or its measurement, was produced from a specific source snapshot.

In our setting, the nodes' vendor or deployer would issue such a CCP for every family member ahead of time; the recursion-theoretic fixed point is still taken over \emph{sources} exactly as in \prettyref{sec:mutual-referencing-programs}.
At runtime, a node receives the peer's $\paren{\mathit{measurement},\mathit{CCP}}$--pair from an \emph{untrusted} third party, such as the peer's deployer or a registry. It reconstructs the peer's source code, checks that the CCP binds \emph{that} source to the measurement, and only then uses the measurement as the reference value.
The node then needs no toolchain and incurs no build latency.

The scheme remains TTP-free in the relevant sense: a forged or mismatched pair is simply rejected, so the CCP provider need not be trusted.
The source code remains the ground truth, and the CCP must show that the supplied measurement was derived from it.
Offloading does, however, require each node to trust the hardware root and endorsement PKI of the TEEs that generates the CCP (referred to here as the \emph{CCP node}).
If the CCP node is based on the same isolation architecture as the other nodes, this requirement adds no new root of trust.
Otherwise, the verifier's trust set expands.
For example, if a CCP node running on an AMD SEV-SNP CVM is used to issue a CCP for an AWS Nitro Enclaves node, one must trust the AMD Root CA in addition to the AWS Nitro Root CA.

Structurally, this variant is closer to DECENT's provisioned credentials \cite{DECENT} than to the fully self-contained construction of \prettyref{sec:nitro}.
We will discuss this further in \prettyref{sec:related}.

\subsection{Updates and redeployment}
\label{sec:updates}

The fixed point is taken over the family as a whole, which has an operational consequence for updates.
In the current implementation of the transpilers, even if a change is made to just one node, it affects the output code and measurements of all members.
An update therefore requires re-transpiling and redeploying the \emph{entire} family.
However, it is also possible to narrow down the scope of the update.

Given a template of nodes, consider the directed graph $G$ with an edge from node $A$ to node $B$ when the body of $A$ references the source of $B$.
Then a change to $B$ can affect $A$ only if $B$ is reachable from $A$, so members with no directed path to $B$ need not change during the update.

This dependency structure can be used to limit the scope of an update.
Let $C\paren{G}$ be the \emph{condensation graph} obtained by contracting each strongly connected component of $G$ to a single vertex.
Note that the condensation graph is always a directed acyclic graph.

A strongly connected component corresponding to a sink (outdegree-$0$ vertex) of $C\paren{G}$ contains no references to nodes outside that component.
Each sink therefore can be transpiled individually by applying \prettyref{thm:simultaneous} only to it.

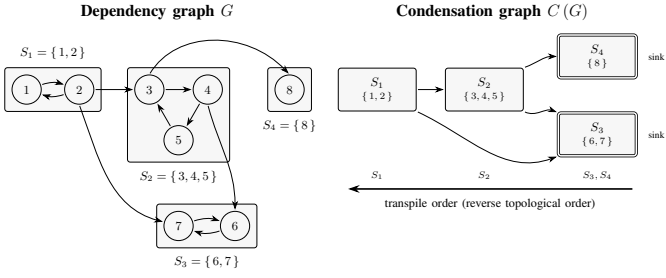
\begin{figure}
\centering
\resizebox{\linewidth}{!}{
\begin{tikzpicture}[
  font=\small,
  vertex/.style={circle, draw, semithick, minimum size=7.5mm, inner sep=0pt},
  edge/.style={-{Stealth[length=2.2mm,width=1.5mm]}, semithick, shorten >=1pt, shorten <=1pt},
  sccbox/.style={draw, rounded corners=2pt, semithick, fill=black!3, inner sep=5pt},
  cnode/.style={draw, rounded corners=2pt, semithick, minimum width=20mm, minimum height=11mm, align=center, fill=black!3},
  sink/.style={cnode, double, double distance=1.1pt},
  title/.style={font=\bfseries\large}
]

\begin{scope}
  \node[title] at (3.35,2.00)
    {Dependency graph $G$};

  \node[vertex] (v1) at (0,0)    {$1$};
  \node[vertex] (v2) at (1.35,0) {$2$};

  \node[vertex] (v3) at (3.15,0)     {$3$};
  \node[vertex] (v4) at (4.65,0)     {$4$};
  \node[vertex] (v5) at (3.90,-1.30) {$5$};

  \node[vertex] (v7) at (3.90,-3.50) {$7$};
  \node[vertex] (v6) at (5.35,-3.50) {$6$};

  \node[vertex] (v8) at (6.75,0) {$8$};

  \begin{scope}[on background layer]
    \node[sccbox,fit=(v1)(v2)]     (scc1) {};
    \node[sccbox,fit=(v3)(v4)(v5)] (scc2) {};
    \node[sccbox,fit=(v6)(v7)]     (scc3) {};
    \node[sccbox,fit=(v8)]         (scc4) {};
  \end{scope}

  \node[anchor=south] at ($(scc1.north)+(0,0.08)$)  {$S_{1}=\set{1,2}$};
  \node[anchor=north] at ($(scc2.south)+(0,-0.08)$) {$S_{2}=\set{3,4,5}$};
  \node[anchor=north] at ($(scc3.south)+(0,-0.08)$) {$S_{3}=\set{6,7}$};
  \node[anchor=north] at ($(scc4.south)+(0,-0.08)$) {$S_{4}=\set{8}$};

  \draw[edge] (v1) to[bend left=18] (v2);
  \draw[edge] (v2) to[bend left=18] (v1);

  \draw[edge] (v3) -- (v4);
  \draw[edge] (v4) -- (v5);
  \draw[edge] (v5) -- (v3);

  \draw[edge] (v6) to[bend left=18] (v7);
  \draw[edge] (v7) to[bend left=18] (v6);

  \draw[edge] (v2) -- (v3);
  \draw[edge] (v2.south) to[out=-70,in=165] (v7);
  \draw[edge] (v4.south) to[out=-75,in=90]  (v6);
  \draw[edge] (v3.north) .. controls +(0.7,1.15) and +(-1.1,1.15) .. (v8.north);
\end{scope}

\begin{scope}[xshift=9.0cm]
  \node[title] at (2.95,2.00)
    {Condensation graph $C\left(G\right)$};

  \node[cnode] (c1) at (0,0)        {$S_1$\\[-1pt]{\scriptsize $\set{1,2}$}};
  \node[cnode] (c2) at (2.75,0)     {$S_2$\\[-1pt]{\scriptsize $\set{3,4,5}$}};
  \node[sink]  (c4) at (5.65,0.85)  {$S_4$\\[-1pt]{\scriptsize $\set{8}$}};
  \node[sink]  (c3) at (5.65,-1.15) {$S_3$\\[-1pt]{\scriptsize $\set{6,7}$}};

  \draw[edge] (c1) -- (c2);

  \draw[edge] (c1.south east) to[out=-25,in=200] (c3.south west);
  \draw[edge] (c2.north east) to[out=20,in=180]  (c4.west);
  \draw[edge] (c2.south east) to[out=-25,in=160] (c3.north west);

  \node[font=\scriptsize, anchor=west] at ($(c4.east)+(0.18,0)$) {sink};
  \node[font=\scriptsize, anchor=west] at ($(c3.east)+(0.18,0)$) {sink};

  \draw[-{Stealth[length=3mm,width=2mm]}, very thick] (6.55,-2.55) -- (-0.75,-2.55) node[midway, below=3pt, align=center] {transpile order (reverse topological order)};

  \node[font=\scriptsize] at (5.65,-2.20) {$S_{3},S_{4}$};
  \node[font=\scriptsize] at (2.75,-2.20) {$S_{2}$};
  \node[font=\scriptsize] at (0,-2.20) {$S_{1}$};

\end{scope}

\end{tikzpicture}%
}

\caption{
Dependency graph $G$ and its condensation graph $C(G)$.
}
\label{fig:condensation-graph}
\end{figure}

More generally, the strongly connected components can be processed in reverse topological order as follows.
Let $S$ be a vertex of $C\paren{G}$ and assume that every vertex reachable from $S$ has already been processed.
Then, references from nodes within $S$ to nodes outside $S$ can be replaced by the source code that has already been generated.
(In the case of mutual attestation, it suffices to hard-code the measurements derived from the generated code.)
The only unresolved references are then those among nodes in $S$, so \prettyref{thm:simultaneous} can be applied solely to $S$.
By repeating this procedure, we obtain the desired code for all nodes (see \prettyref{fig:condensation-graph}).
Formally, this procedure can be described as follows:

\begin{thm}
Let $\varphi^{\paren{n}}$ be an acceptable numbering of $\mathcal{P}^{\paren{n}}$ and let $f_{1}, \dots, f_{s+t} \colon \mathbb{N}^{s} \times \mathbb{N}^{t} \times \mathbb{N}^{n} \rightharpoonup \mathbb{N}$ be computable partial functions.
Assume that for all $1\leq i\leq s$, the function $f_{s}$ does not dependent on the variables in the second block $\mathbb{N}^{t}$.
Fix a solution $\vec{e}_{S}:=\paren{e_{1},\ldots,e_{s}}\in \mathbb{N}^s$ to the following system of equations:
\[
\varphi^{\paren{n}}_{e_{i}} = f_{i}\paren{ \vec{e}_{S}, 0, \ldots, 0, {-} }, \qquad 1 \leq i \leq s.
\]
Consider the functions $g_{1},\ldots,g_{t} \colon \mathbb{N}^{t} \times \mathbb{N}^{n} \rightharpoonup \mathbb{N}$ defined by $g_{j} := f_{s+j}\paren{\vec{e}_{S}, {-}, {-}}$, and
let $\vec{e}_{T}:=\paren{e_{s+1}, \ldots, e_{s+t}} \in \mathbb{N}^{t}$ be an arbitrary solution to
\[
\varphi^{\paren{n}}_{e_{s+j}} = g_{s+j}\paren{ \vec{e}_{T}, {-} }, \qquad 1 \leq j \leq t.
\]
Then the $\paren{s+t}$--tuple $\vec{e} := \paren{e_{1},\ldots,e_{s},r_{1},\ldots,r_{t}}$ gives a solution to
\[
\varphi^{\paren{n}}_{e_{i}} = f_{i}\paren{\vec{e}, {-}},\qquad 1\leq i\leq s+1.
\]
\end{thm}

\begin{proof}
For $1\leq i \leq s$, since $f_{i}$ does not depend on the second block by assumption, we have
\[
\varphi^{\paren{n}}_{e_{i}} = f_{i}\paren{ \vec{e}_{S}, 0, \ldots, 0, {-} } = f_{i}\paren{ \vec{e}_{S}, \vec{e}_{T}, {-} }.
\]
For $1 \leq j \leq t$, we have
\[
\varphi^{\paren{n}}_{e_{s+j}} = g_{j}\paren{ \vec{e}_{T}, {-} } = f_{s+j}\paren{e_{S},e_{T}}.\qedhere
\]
\end{proof}

According to this theorem, to solve a system of fixed-point equations with $s+1$ variables, one need only follow these two steps:
\begin{enumerate}
\item Solve the system of fixed-point equations with $s$ variables obtained by fixing $t$ variables.
\item Next, solve the system of fixed-point equations with $t$ variables obtained by substituting the solution from (1).
\end{enumerate}
This decomposition can apply repeatedly.
Here, the set $\set{1, \ldots, s}$ corresponds to $S$ in the previous informal explanation.

Under this refinement, even if the body of a node $v$ is altered, the effect is limited to the strongly connected components reachable to $v$ in $C\paren{G}$.
Only those components require recompilation and redeployment.
No action is required for components that have no path to $S$.

\section{Related work}
\label{sec:related}

In this section, we compare our approach with existing approaches that address or avoid the same reference-value bootstrapping problem.

\subsection{Trusted third parties (TTPs)}
\label{sec:notary}

The most direct way to avoid this problem is to move the reference values out of the nodes' measured artifacts.
An external trusted third party (in RATS terms, a Reference Value Provider for the whole group) authorises the reference values, and nodes obtain them at runtime.
Since no reference value appears inside any node's measured code, the recursive dependency never arises.
The price is trust in the provider, which can authorise the reference value of arbitrary code.

A guide to AWS Nitro Enclaves \cite{AWSBlog} states this bootstrapping problem directly and proposes a \emph{Measurement Notary Service} as a remedy.
An enclave measurement is approved by $K$ of $N$ authorised signers and stored in the notary's database, and enclaves fetch reference values from this database at runtime.
This is precisely the type of TTP that our construction eliminates.

Apache Teaclave \cite{Teaclave} addresses the same problem among its SGX enclave services.
Each enclave hard-codes only the public keys of trusted auditors, and loads peer measurements together with the auditors' signatures at runtime.
This removes the need for an online notary, but the auditors remain trusted authorities over which enclave identities are acceptable.

Our construction requires neither an online notary nor auditor-signed reference values, since each node derives the reference measurements solely from intrinsic data.

\subsection{Certificate-bound authorisation lists}

The Decent Application Platform \cite{DECENT} supports mutual attestation among Intel SGX enclaves without a trusted party beyond Intel.
At launch, every node loads an \emph{authorisation list} (AuthList) that contains the reference values of all nodes in the application instance.
The AuthList is loaded as data after \MRENCLAVE{} has been fixed, so mutually dependent entries change no node's identity.
Each node uses its AuthList as the reference values for mutual attestation, and the AuthList is bound to the node's attestation credential.
A node accepts a peer only if the peer's code digest is authorised in its AuthList and the peer's bound AuthList is identical to its own.
Clients hold the list they trust and connect only to instances that carry it.

DECENT thus provisions peer reference values from outside the node and checks only that peers agree on them, so its guarantee is external in the sense of \prettyref{sec:intrinsic}.
Its mechanism is portable to any TEE whose evidence can carry a bound AuthList, although the prototype architecture is specific to SGX, whereas our construction requires only the architecture's measurement function.

\subsection{Hardware-rooted reference-value injection}
\label{sec:injection}

Some TEEs provide launch-time configuration fields that do not affect the code measurement, become immutable once the TEE has been launched, and are included in its attestation evidence.
Such a field can be used to provision peer reference values without creating a circular dependency.
For mutual attestation, the peer reference values are written to this field at launch time.
After launch, the node obtains its own attestation evidence, reads the values from the field, and uses them as reference values to appraise its peers.
If the field is too small to contain all peer reference values, it may instead contain a digest of a reference-value list supplied separately at runtime.
The node uses the list only after verifying its digest against the value bound to its own attestation evidence.

For Intel SGX, Momose, Qin, Sakurai, and Vij~\cite{KSS} use the \emph{Key Separation and Sharing (KSS)} feature in two ways.
A fresh worker identifier bound to \CONFIGID{} lets the attester reject reused worker enclaves.
To resolve the mutual-attestation cycle, the attester hard-codes the worker's measurement, while the worker is launched with the attester's measurement bound to its own \CONFIGID{}.
The worker uses this value to verify the attester, and the attester checks the same value in the worker's quote.
Because \CONFIGID{} is attested but does not change \MRENCLAVE{}, this launch-time injection avoids a circular dependency.

Confidential VMs offer analogous fields populated by the host at launch time.
Intel TDX provides \texttt{MR\_CONFIG\_ID}, while AMD SEV-SNP provides \texttt{HOST\_DATA}, so the host can implement the same scheme.
These fields are fixed at launch, so the host or launcher must supply the reference values before the node starts.
In Confidential Containers, these fields are also used to bind the container policy supplied by the host with the attestation evidence.

This approach allows the same measured artifact to be deployed under different peer policies by provisioning different sets of peer reference values.
However, it relies on architecture-specific attested configuration fields and requires the host or launcher to provision the intended reference values correctly at deployment.

Because arbitrary reference values can be supplied from outside, a node cannot inspect the values it is given or judge what code they denote.
Its guarantee is therefore external in the sense of \prettyref{sec:intrinsic}, with the host or launcher in the role of the provider.

\subsection{Measurement-specific reference-value derivation}
\label{sec:MAGE}

\subsection{Measurement-specific reference-value derivation}

MAGE \cite{MAGE} addresses the same problem directly and describes the recursive dependency in the same terms.
For mutual attestation, two enclaves would each need the other's identity (\MRENCLAVE{}) in its own initial data, yet embedding it changes that data and hence both identities.
MAGE breaks the cycle without a trusted third party by instrumenting the enclaves so that each \emph{derives} the others' identities from its own initial data.
Every enclave carries a reserved section (\MARS{}) whose content is identical across the group.
For each member, this section records the \emph{pre-measurement}, the \MRENCLAVE{} hash state of all pages preceding \MARS{}, together with the loading parameters needed to continue the hash (collectively \MAINFO{}).
Since SGX extends \MRENCLAVE{} over pages in load order, an enclave that loads \MARS{} last can complete any member's hash from that member's pre-measurement and the content of its own \MARS{}.

The prototype implements this layout by extending the Intel SGX SDK with a library that reserves \MARS{} and exposes derivation APIs, a signing tool that extracts \MAINFO{} and populates \MARS{}, and an enclave loader that loads \MARS{} last \cite[Section~5.1]{MAGE}.
MAGE also sketches a variant that keeps the standard loader by moving the content after \MARS{} to untrusted storage and recording its hash in the common part \cite[Section~7.1]{MAGE}.
The signing tool and the library remain necessary in either case.

Structurally, MAGE is a measurement-level analogue of our construction.
Both place a byte-identical aggregate in every member and select a member by index at runtime, which is the same shape as \pyreflect{}'s implementation of \prettyref{thm:simultaneous} (\prettyref{app:design}).
The difference lies in what the aggregate records.
MAGE stores one-way intermediate hash states, so members can reproduce one another's measurements but not their code.
Its cycle is resolved by decomposing the measurement function so that a fixed-size state of the member-specific part and the shared part together determine the identity, which in the SGX instantiation means measuring \MARS{} last and in the Merkle-tree variant means placing it as a sibling subtree \cite[Appendix~A]{MAGE}.
Our aggregate stores source, and the cycle is resolved by Kleene's recursion theorem independently of how the measurement is computed.

This difference cuts both ways.
Because MAGE never exposes source, a member with proprietary code can participate by contributing only its pre-measurement and size \cite[Section~7.4]{MAGE}, whereas our construction embeds every member's source in the family.
Conversely, MAGE is intrinsic in the sense of \prettyref{sec:intrinsic} only at the level of measurements.
A node derives peer identities from its own initial data, but it cannot determine from within what code a derived identity denotes, so auditing requires the peers' code from outside.
Both approaches require all members' code at build time and neither supports updating a single member without rebuilding the group \cite[Sections~3.3 and~7.3]{MAGE}.
We have discussed this issue and its mitigation in \prettyref{sec:updates}.

The technique is tied to the measurement process.
MAGE argues that it extends to other TEEs with measurement-based identities, discussing AMD SEV/SEV-SNP and ARM TrustZone without detailed procedures \cite[Section~7.2]{MAGE}, and gives the Merkle-tree construction for TEEs whose measurement is a tree hash.
Each instantiation must nevertheless expose a compatible intermediate state and control where the shared part enters the measurement, which in practice means adapting the SDK, image-construction or signing tool, and loader of that TEE.
A heterogeneous deployment must implement every participating TEE's measurement-specific path.
Our fixed-point construction instead operates unchanged at the source level for any measurement function and uses the architecture's existing build and measurement interfaces without modifying its SDK, signing tools, loader, or attestation stack (\prettyref{sec:tpm}, \prettyref{sec:nitro}).
Both MAGE and our construction are TTP-free, so the distinction is architectural rather than one of additional trust.

\subsection{Common-image deployments}
\label{sec:common-image}

Yet another way to avoid the reference-value bootstrapping problem is available when every node runs the same program \cite{BMW16}.

Consider first the unary bootstrapping requirement:
\[
\varphi_{e} = f\paren{\mu \circ \beta\paren{e}, {-}},
\]
where $e$ is the common source code, $\beta\paren{e}$ is its deployed artifact, and $\mu \circ \beta\paren{e}$ is the common measurement.
Unlike the mutually recursive case, this dependency can be resolved directly by the ordinary attestation mechanism.
A node can request attestation evidence for itself and read the measurement contained in that evidence.
Provided that the evidence is securely associated with the requesting node, the resulting value is precisely $\mu \circ \beta\paren{e}$.
Since every node runs the same deployed artifact, this value can also be used as the reference measurement for every peer.

A single common program can support different node roles by combining all role definitions.
An externally provisioned role parameter selects the appropriate mode at launch time or runtime.
The selected role can be bound to the attestation evidence through report data or an analogous field (such as qualifying data).
A peer can then verify that the node runs the common measured image and that it claims the expected role.
The program must ensure that the selected role remains consistent with the role bound to the evidence.

Binding the role parameter $\rho$ to the evidence associates the claimed role with the attested node.
By itself, however, this shows only that the peer has the same measurement $m$ and reports $\rho$.
To conclude that the peer actually behaves as role $\rho$, the verifier must also know from the source or build artifact corresponding to $m$ that the program keeps $\rho$ immutable, uses it to select the enabled behaviour, and reports it faithfully.
Our construction instead derives the specific expected measurement of each peer from its reconstructed source code.

At the level of reference-value bootstrapping, the common-image approach is architecture-independent and TTP-free.
It therefore avoids the mutual dependency by collapsing all nodes into a single measured identity, rather than by constructing distinct measured identities with mutually consistent reference values.

The common-image approach has a practical cost when different roles implement different functionality.
Every node must include the code and dependencies for all roles, including roles that it never uses.
For a build-artifact deployment, all of this code must be compiled or linked into every deployed artifact.
The resulting artifacts may therefore be unnecessarily large, and a change to any role requires rebuilding and redeploying the common image.
This cost can be reduced using the approach described in \prettyref{sec:updates}.

Our construction likewise makes information about the entire system available to each node, but as a source-level representation.
The definitions of the other nodes are embedded as data in the common blob, while their executable implementations need not be compiled or linked into the node's own role implementation.
In a build-artifact deployment, a node may still carry the peer build closures and the toolchain needed to build them.
The rebuilt peer artifacts and build closures, however, are temporary and may be discarded once their reference measurements have been derived.

\section{Conclusion}

We formulated the reference-value bootstrapping problem for mutual attestation as a system of simultaneous fixed-point equations, and solved it effectively using Kleene's second recursion theorem.
The resulting construction enables each node to reconstruct its peers' source code at runtime from intrinsic data alone and to derive their reference values from that source.
The guarantees it provides are `intrinsic' in the sense that every input to this derivation lies within the node's own measured code.
A node can therefore check its own deployment against its attestation evidence, confirm from the peers' reconstructed source that they appraise it in turn, and thus verify from within the measured artifacts that the nodes mutually attest.

The derivation from source to reference value depends on the deployment model.
In a source-code deployment, a node applies the measurement procedure (typically, one-shot application of a cryptographic hash function) directly to the reconstructed source.
In a build-artifact deployment, it first rebuilds the peer's artifact at runtime and then applies the measurement procedure.
Here, the build must be reproducible so that the source uniquely determines the resulting measurement.

We implemented the \pyreflect{} and \nixreflect{} transpilers, which automate the construction underlying Kleene's theorem for Python and Nix, respectively.
Using them, we developed two proofs of concept for the two deployment models.
One involves two Python servers equipped with TPMs performing post-handshake mutual attestation.
The other is a quine-like system in which two Nitro Enclaves display each other's reference PCR values.

For build-artifact deployments, the runtime cost of deriving reference values is dominated by the build of peer artifacts.
We showed how this computation can be offloaded to attestable builds without introducing a trusted provisioning party, although this may add the CCP's hardware root to the trust base.

With the current implementation of the transpilers, the entire system must be re-transpiled and redeployed, even when updating a single node.
We showed that this issue can be mitigated by performing incremental transpilation along with topological order.

To the best of our knowledge, our approach is the first one that requires no trusted third party or external reference-value provider, is architecture-independent, provides intrinsic guarantees in the sense of \prettyref{sec:intrinsic}, and allows different nodes to retain distinct behaviour definitions and measured identities.

More broadly, this work turns a classical theorem of mathematical logic into a practical design principle for modern systems security.
Kleene proved the second recursion theorem in 1938 \cite{Kleene_1938}, long before any of the systems considered here existed.
Its relevance to computing has nevertheless been recognised for decades.
The theorem underlies von Neumann's self-reproducing automata \cite{vonNeumann1966, Arb90}, and it provides, in a precise sense, the theoretical basis of computer viruses and worms, where a virus is obtained as the fixed point of its own propagation function \cite{Viruses,AbstractViruses}.
Our construction places mutual attestation in the same line.
The recursion theorem yields programs that carry a faithful description of themselves and of one another, and confidential computing turns this self-description into a verifiable claim.
Revisiting classical results of this kind remains a fruitful source of designs for contemporary problems.

\section*{Acknowledgement}

Generative AI tools (OpenAI Codex and Anthropic Claude Code) were used in the preparation of this work.
The tools were used to develop the proof-of-concept implementations described in \prettyref{sec:transpiler}, \prettyref{sec:tpm}, and \prettyref{sec:nitro}.
All generated code was reviewed, tested, and revised by the author.
The tools were also used to assist in the execution time benchmarks in \prettyref{sec:cost}.
All benchmarks and their results were obtained by actually running the test code.
The author reviewed, revised, and approved all AI-assisted content and take full responsibility for the content of this paper.

\bibliographystyle{IEEEtran}
\bibliography{mutual-ra}

\appendix

\subsection{Code availability.}
\label{sec:code-availability}

All PoC implementations are available at \url{https://github.com/acompany-develop/mutual-attestation-kleene}.

\subsection{Proofs of Kleene's recursion theorems}
\label{app:proofs}

Both \prettyref{thm:kleene} and \prettyref{thm:simultaneous} are classical results of recursion theory.
However, since the transpilers in \prettyref{sec:transpiler} are based on their proofs, we recall the proofs here.
The correspondence with the transpiler is described in \prettyref{app:design}.

\begin{proof}[Proof of \prettyref{thm:kleene}]
For $n,m \in\mathbb{N}$, let $S^{m}_{n}\colon \mathbb{N}\times\mathbb{N}^{m} \to \mathbb{N}$ be a computable function such that
\[
\varphi^{\paren{n}}_{S^{m}_{n}\paren{ e, \vec{y} }} = \varphi^{\paren{n+m}}_{e}\paren{{-}, \vec{y}}.
\]
Note that such a function always exists by the $S^{m}_{n}$ theorem \cite{Kleene_1938}.
Consider the following computable partial function $g \colon \mathbb{N}^{n} \times \mathbb{N} \times \mathbb{N} \rightharpoonup \mathbb{N}$:
\[
g\paren{\vec{x}, y, z} := \varphi^{\paren{1+n}}_{z}\paren{ S^{1}_{n}\paren{ y, y }, \vec{x} }.
\]
Let $j$ be an index of $g$ with respect to $\varphi^{\paren{n+2}}$.
The computable function $e\colon \mathbb{N} \to \mathbb{N}$ defined by
\[
e\paren{i} := S^{1}_{n}\paren{S^{1}_{n+1}\paren{j, i}, S^{1}_{n+1}\paren{j, i}}
\]
satisfies the equation
\[
\varphi^{\paren{n}}_{e\paren{i}}
= \varphi^{\paren{n+1}}_{i}\paren{e\paren{i}, {-}}.
\]
Consequently, choosing $i$ to be any $\varphi^{\paren{n+1}}$-index of $f$ gives
$\varphi^{\paren{n}}_{e\paren{i}} = f\paren{e\paren{i}, {-}}$. \qedhere
\end{proof}

\begin{proof}[Proof of \prettyref{thm:simultaneous}]
The simultaneous form can be reduced to the case $k = 1$ (\prettyref{thm:kleene}) by using the computable isomorphism $\mathbb{N}^{k} \cong \mathbb{N}$.
To clarify the connection with the implementation of our transpilers, we also give a more concrete proof.

Define the computable partial function $F \colon \mathbb{N} \times \mathbb{N}^{n} \times \mathbb{N} \rightharpoonup \mathbb{N}$ by
\[
F\paren{ y, \vec{x}, i } := \begin{cases}
f_{i}\paren{S^{1}_{n}\paren{y, 1}, \ldots, S^{1}_{n}\paren{y, k}, \vec{x}}, & 1\leq i\leq k, \\
\text{undefined}, & \text{otherwise}.
\end{cases}
\]
\prettyref{thm:kleene} provides a $\varphi^{\paren{n+1}}$--index $e \in \mathbb{N}$ such that
\[
\varphi^{\paren{n+1}}_{e}\paren{\vec{x}, i} = F\paren{e, \vec{x}, i}.
\]
For each $1\leq i \leq k$, we have
\begin{align*}
\varphi^{\paren{n}}_{S^{1}_{n}\paren{e, i}} \paren{\vec{x}}
&= \varphi^{\paren{n+1}}_{e}\paren{\vec{x}, i} \\ 
&= F\paren{e, \vec{x}, i} \\ 
&= f_{i}\paren{S^{1}_{n}\paren{e, 1}, \ldots, S^{1}_{n}\paren{e, k}, \vec{x}}.
\end{align*}
Hence the tuple $\vec{e} := \paren{S^{1}_{n}\paren{e, 1}, \ldots, S^{1}_{n}\paren{e, k}}$ gives a solution to the system of equations:
\[
\varphi^{\paren{n}}_{e_{i}} = f_{i}\paren{ \vec{e}, {-} }, \qquad 1 \leq i \leq k.\qedhere
\]
\end{proof}

\subsection{Transpiler interface and emitted nodes}
\label{app:design}

This appendix specifies the interface implemented by \pyreflect{} and then describes the common data and code used to implement it.
\Nixreflect{} uses the same fixed-point idea for Nix expressions, with a language-specific wrapper.

\emph{Interface.}
A template is a JSON list with one entry for each node:
\begin{verbatim}
[
  {
    "node-id": "__NODE1",
    "code":    "... __NODE2 ..."
  },
  {
    "node-id": "__NODE2",
    "code":    "... __NODE1 ..."
  }
]
\end{verbatim}
Let $\mathcal{N}=\set{N_{1},\ldots,N_{n}}$ be the node-ids and let $P_{N}$ be the Python body associated with $N$.
For every $N\in\mathcal{N}$, the transpiler emits a self-contained Python source file $e_{N}$.
It rewrites each designated occurrence of a node-id $M$ in $P_{N}$ into a call to the \render{} function with argument $M$.
The interface guarantee is
\[
\operatorname{render}_{e_{N}}\paren{M}=e_{M},
\qquad N,M\in\mathcal{N},
\]
where the right-hand side denotes the exact source text of the emitted file for $M$.
Consequently, $e_{N}$ behaves as $P_{N}$ with every node-id interpreted as the source of the corresponding output file.

\emph{Generated representation.}
Write $\widehat{P}_{N}$ for the body $P_{N}$ in which each node-id occurence has been replaced by a call to \render{}.
The transpiler constructs the canonical record
\[
\mathcal{D} = \braket{\mathcal{N}, \set{N\mapsto\widehat{P}_{N}}, H, F},
\]
whose Base64-encoded canonical JSON representation $d$ will be embedded into every output.
Here $H$ is a fixed Python header, and $F$ is the \emph{reconstruction framework}: fixed helper code that decodes $d$ and defines \render{}.
The framework is separate from the user-supplied body.

Using $\parallel$ for textual concatenation, the emitted source for $N$ is
\[
e_{N} := H \parallel \operatorname{SELF}\paren{N} \parallel \operatorname{DATA}\paren{d} \parallel F \parallel \widehat{P}_{N}.
\]
The term $\operatorname{SELF}\paren{N}$ denotes the one-line Python assignment that binds \self{} to the node-id $N$, while $\operatorname{DATA}\paren{d}$ binds \data{} to the embedded representation $d$.
The components $H$, $\operatorname{DATA}\paren{d}$ and $F$ are identical in every emitted source.
Only the components $\operatorname{SELF}\paren{N}$ and $\widehat{P}_{N}$ depend on $N$.

\emph{Reconstruction.}
When evaluating \render{} with $M$, the framework $F$ decodes $d$, selects $\widehat{P}_{M}$ from the decoded JSON, and reconstruct the layout of $e_{M}$:
\[
e_{M} = H \parallel \operatorname{SELF}\paren{M} \parallel \operatorname{DATA}\paren{d} \parallel F \parallel \widehat{P}_{M}.
\]
\prettyref{algo:pseudocode} describes the instructions for $e_{N}$.

\begin{algorithm}[t]
\caption{Pseudocode of $e_{N}$.}
\label{algo:pseudocode}

\begin{algorithmic}[1]
\State $H$
\State $\mathtt{SELF} \gets N$
\State $\mathtt{DATA} \gets d$

\Function {render}{$M$}
  \State $\mathcal{D} \gets \Call{decode}{\mathtt{DATA}}$
  \State $H \gets \mathcal{D}\sqbraket{\mathtt{header}}$
  \State $F \gets \mathcal{D}\sqbraket{\mathtt{framework}}$
  \State $\widehat{P}_{M} \gets \mathcal{D}\sqbraket{\mathtt{bodies}}\sqbraket{M}$
  \State \Return $H \parallel \mathtt{SELF} \leftarrow M \parallel \mathtt{DATA} \leftarrow d \parallel F \parallel \widehat{P}_{M}$
\EndFunction

\State $\widehat{P}_{N}$
\end{algorithmic}
\end{algorithm}

The shared record $\mathcal{D}$ packages the family of all nodes into one object and the target node-id selects one member, matching the reduction in \prettyref{thm:simultaneous}.
Re-emitting the current data value as its own source literal implements the quine-style self-application used in \prettyref{thm:kleene}.
Together these establish the interface guarantee.

\prettyref{fig:layout} also clarifies the comparison with MAGE in \prettyref{sec:related}: both designs combine shared group data with a member selector, but \pyreflect{} reconstructs source whereas MAGE completes an SGX measurement from intermediate hash states.

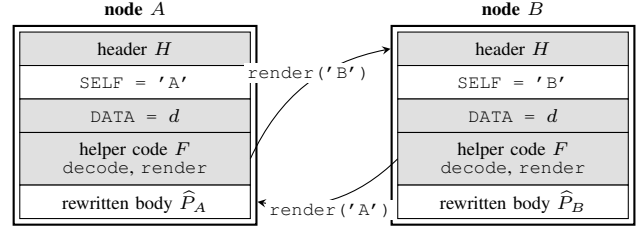
\begin{figure}
\centering
\begin{tikzpicture}[font=\scriptsize, >=stealth,
  part/.style={draw, minimum width=3.05cm, minimum height=2.8ex,
               inner sep=1.5pt, align=center, outer sep=0pt},
  common/.style={part, fill=black!12},
  local/.style={part, fill=white}]
\node[common]                    (a1) at (0,0)       {header $H$};
\node[local, anchor=north]       (a2) at (a1.south)  {\texttt{SELF = 'A'}};
\node[common, anchor=north]      (a3) at (a2.south)  {\texttt{DATA = $d$}};
\node[common, anchor=north,
      minimum height=4.4ex]      (a4) at (a3.south)  {helper code $F$\\[-1pt]\texttt{decode}, \texttt{render}};
\node[local, anchor=north]       (a5) at (a4.south)  {rewritten body $\widehat{P}_{A}$};
\node[draw, thick, inner sep=2.3pt, fit=(a1)(a5),
      label={[font=\scriptsize\bfseries]above:node $A$}] (fA) {};
\node[common]                    (b1) at (5.0,0)       {header $H$};
\node[local, anchor=north]       (b2) at (b1.south)    {\texttt{SELF = 'B'}};
\node[common, anchor=north]      (b3) at (b2.south)    {\texttt{DATA = $d$}};
\node[common, anchor=north,
      minimum height=4.4ex]      (b4) at (b3.south)    {helper code $F$\\[-1pt]\texttt{decode}, \texttt{render}};
\node[local, anchor=north]       (b5) at (b4.south)    {rewritten body $\widehat{P}_{B}$};
\node[draw, thick, inner sep=2.3pt, fit=(b1)(b5),
      label={[font=\scriptsize\bfseries]above:node $B$}] (fB) {};
\draw[->] (a4.east) to[bend left=27]
  node[above, pos=0.52, fill=white, inner sep=1pt] {\texttt{render('B')}} (fB.west |- b1);
\draw[->] (b4.west) to[bend left=27]
  node[below, pos=0.52, fill=white, inner sep=1pt] {\texttt{render('A')}} (fA.east |- a5);
\end{tikzpicture}
\caption{Two emitted files. Shaded components are common to the family; \texttt{render} reconstructs the complete source selected by its node-id.}
\label{fig:layout}
\end{figure}

\subsection{Benchmark environment}
\label{app:bench-env}

\prettyref{tab:bench-env-full} records the host's configuration used for the benchmark in \prettyref{sec:cost}.

\begin{table}
\caption{Host environment for the Nitro Enclaves cost experiment.}
\label{tab:bench-env-full}
\small
\begin{tabularx}{\columnwidth}{@{} l >{\raggedright\arraybackslash}X @{}}
\toprule
Item & Value \\
\midrule
Instance type & \texttt{m6g.xlarge}: AWS Graviton2 processor (Arm Neoverse-N1), 4~vCPUs, 16~GiB RAM \\
AMI & \nolinkurl{ubuntu/images/hvm-ssd-gp3/ubuntu-resolute-26.04-arm64-server-20260604} (\texttt{ami-04d0f56e9ce314a8e}, \texttt{us-east-2}) \\
OS / kernel & Ubuntu 26.04 LTS / Linux 7.0.0-1010-aws \\
Host storage & $64$\,GiB \\
Nitro Enclaves feature & enabled \\
Allocator pool & \texttt{memory\_mib: 8192}, \texttt{cpu\_count: 2} \\
Per-enclave allocation & \texttt{--memory 2048 --cpu-count 2} \\
\Nixreflect{} & version 0.1.0 (commit ID: 676fce9) \\
\nitroutil{} & master branch (commit ID: b529ed6) \\
\nitrocli{} & version 1.4.5 (commit ID: 18a5f6f) \\
Docker & version 29.1.3; Ubuntu package \texttt{docker.io} 29.1.3-0ubuntu4.1 \\
Nix & Determinate Nix 3.21.8 (Nix 2.34.8) \\
\bottomrule
\end{tabularx}
\end{table}

\end{document}